\documentclass[onecolumn]{IEEEtran}
\usepackage{color,latexsym,amsfonts,amssymb}
\usepackage{amsmath,cite}
\usepackage{amsthm}
\usepackage[mathscr]{euscript}

\usepackage{authblk}

\ifCLASSINFOpdf 
\usepackage[pdftex,colorlinks, 
           bookmarks=true,%
 pdftitle=Stronger\ wireless\ signals\ appear\ more\ Poisson,%
           pdfauthor=B.\ Blaszczyszyn%
           \ -\ P.\ Keeler%
           \ - N. Ross%
           \ - A. Xia,]{hyperref}  
\usepackage{hypernat} 
\hypersetup{
   bookmarksnumbered,
   pdfstartview={FitH},
   citecolor={blue},
   linkcolor={red},
   urlcolor=[rgb]{0,0.55,0},
   pdfpagemode={UseOutlines}}
\else

 \usepackage[dvips]{graphicx} 
 \usepackage[numbers,sort&compress]{natbib} 
\fi

\newcommand{\E}{\mathbb{E}} 

\newcommand{\R}{\mathbb{R}}
\newcommand{\Prob}{\mathbb{P}}
\newcommand{\ind}{\mathbf{1}}
\newcommand{\Ind}{\mathbf{1}}
\newcommand{\mean}{{\E}}
\def\tt{\tau}

\def\tup{t'\uparrow}

\def\fade{{S}} 
\def\shadow{{S}} 
\def\rayleigh{{S}} 

\def\Measure{{M}} 

\def\dtv{d_{\textrm{TV}}}
\def\law{\mathcal{L}}

\def\t#1{^{(#1)}}
\def\res#1{|_{#1}}

\def\eq{\begin{equation}}
\def\en{\end{equation}}
\def\convp{\stackrel{{{\pr}}}{\longrightarrow}}
\def\pr{\mathbb{P}}

\newtheorem{Theorem}{Theorem}

\theoremstyle{definition}

\newcommand{\OneOrTwoColumnDisplay}[2]{#1} 

\title{Stronger wireless signals appear more Poisson}

\begin{document}

\date{\today}


\author{Paul Keeler$^*$\thanks{$^*$Weierstrass Institute, 10117 Berlin, Germany; keeler@wias-berlin.de},  
Nathan Ross$^\dagger$, Aihua Xia$^\dagger$\thanks{$^\dagger$University of Melbourne, Parkville, VIC 3010, Australia; \{aihua.xia,\,nathan.ross\}@unimelb.edu.au}
and Bart{\l }omiej~B{\l }aszczyszyn$^\ddagger$ \thanks{$^\ddagger$Inria-ENS, 75589 Paris, France; Bartek.Blaszczyszyn@ens.fr}}

\maketitle

\begin{abstract}
Keeler, Ross and Xia~\cite{keeler2016wireless} recently derived approximation and convergence results, which imply that the point process formed from the signal strengths received by an observer in a wireless network under a general statistical propagation model can be modelled by an inhomogeneous Poisson point process on the positive real line. The basic requirement for the results to apply is that there must be a large number of transmitters with different locations and random propagation effects.
The aim of this note is to apply some of the main results of~\cite{keeler2016wireless} in a less general but more easily applicable form to illustrate how the results can be applied in practice. New results are derived that show that it is the strongest signals, after being weakened by random propagation effects, that behave like a Poisson process, which supports recent experimental work. 
\end{abstract}

\begin{keywords}
stochastic geometry, propagation model,  error bounds, Poisson approximation
\end{keywords}

\section{Introduction} 
Standard assumptions for stochastic geometry models of wireless networks are that transmitter positions form a homogeneous Poisson point process, and each transmitter emits a signal whose power or strength is altered deterministically by \emph{path loss} distance effects and randomly by \emph{propagation effects}  (due to signal fading, shadowing, varying antenna gains, etc., or some combination of these). The Poisson assumption is convenient, particularly under the power-law path-loss model, because  many fundamental quantities become analytically tractable such as the signal-to-interference ratio~\cite{blaszczyszyn2014studying}, which is a simple transform of the two-parameter Poisson-Dirichlet process~\cite{sinrtwopd2014}. 


But in  wireless networks the transmitter positions do not always resemble realizations of a Poisson process.  However, even if a network does not appear Poisson,  recent results show that in the presence of sufficient propagation effects, the signal strengths observed by a single observer are close to those under the Poisson assumption. That is, for a single observer in the network, the presence of random propagation effects, which weaken the signals, renders the process of the incoming signal strengths close to a Poisson point process. For log-normal shadowing, this was confirmed by deriving a convergence theorem for  log-normal-shadowing-based (for example, Suzuki) models~\cite{blaszczyszyn2014wireless}. Keeler, Ross and Xia~\cite{keeler2016wireless} extended these Poisson convergence results considerably to the case of any reasonable path-loss model and a general class of probability distributions for propagation effects. They also derived bounds on the distance between the distributions of processes of signal strengths resulting from Poisson and non-Poisson transmitter configurations. These bounds compare the two point processes of signal strengths by matching their intensity measures and so may form the basis of methods to  experimentally confirm (or reject) whether signal strengths can be modeled with a Poisson process. 
The aim of this note is to apply these results in a more approachable way and to derive new  results.

\section{Model}
We give general conditions for the propagation model and the transmitter configuration, but stress that  some of these assumptions can be more general. We have chosen them to reflect popular model choices and to reduce mathematical technicalities, and refer the reader to~\cite{keeler2016wireless}. Let a sequence of positive i.i.d.\ random variables $ \fade_1, \fade_2,\dots$ represent propagation effects such as fading, shadowing, and so on, and let $\fade$ be an independent copy of any $\fade_i$. Assume $\ell(x)$ is a non-negative function of the form $\ell(x)=1/h(|x|)$, where $h$ is positive on the positive real line $\R_+^0:=(0,\infty)$, left-continuous and nondecreasing with generalized inverse $h^{-1}(y)=\inf\{x:h(x)>y\} $.   We assume an observer or user is located in the network at the origin and the transmitters are located according to a locally finite deterministic point pattern $\phi=\{x_i\}_{i\geq 1}$ on $\R^2\setminus\{0\}$, where the origin is removed to stop the observer being located on a transmitter, which would prevent Poisson convergence. $\phi$ could also be a single \emph{realization} of a random point process $\Phi$. If the transmitters do form a point process $\Phi$, then  our results  continue to hold with some embellishments. 

For the general propagation model, we define the signal powers or strengths at the origin $P_i:=\fade_i \ell(x_i)$ emanating from 
a transmitter located at $x_i\in \phi$, $i=1,2,\ldots$. We are interested in the point process formed by these power values 
$
 \Pi:=\{ P_1, P_2, \dots \} ,
$
but in anticipation of an infinite cluster of points at the origin, we consider point process formed from 
the inverse values
$
N:=\{ V_1,V_2, \dots \},
$
where $V_i:=1/P_i$, which is  called the \emph{propagation process} or the \emph{path-loss with fading process}.

\section{Preliminaries}
\subsection{Intensity measure of signals}
Let  $\phi(r):=\phi(B_0(r))$ or $\Phi(r):=\Phi(B_0(r))$ denote the number of points of $\phi$ or $\Phi$ located in a disk or ball $B_0(r)$ centered at the origin with radius $r$. For a deterministic point pattern of transmitters $\phi$, the intensity measure $\Measure$ of the point process $N$ satisfies
\begin{equation}\label{intensity1}
\Measure((0,t])=\E[\sum_{V_i\in N}\Ind(V_i\leq t)]=\E [\phi (  h^{-1} (t \fade))],
\end{equation}
where $\Ind$ is an indicator function. $\Measure(t):=\Measure((0,t])$ is the expected number of points of $\phi$ located in a disk of (random) radius $h^{-1}(t\fade)$. If the transmitters form a random point process $\Phi$ with intensity measure determined by $\Lambda((0,r])=\E [\Phi (r )]=:\Lambda(r)$, then the  intensity measure of $N$ satisfies
\begin{equation}\label{intensity2}
\Measure(t) =\E[\Lambda(  h^{-1} (t\fade))].
\end{equation}
 This result holds for any point process $\Phi$, see~\cite[Propositions 2.6 and 2.9]{keeler2016wireless}, and has practical applications such as statistically fitting models; see Section~\ref{s.estimateM}. If $\Phi$ is a stationary point process with intensity $d\Lambda(r)=2\pi\lambda  rdr$ (so the density of points is $\lambda$), then 
\begin{equation}\label{intensity2a}
\Measure(t) =\pi\lambda \E[(h^{-1} (t\fade))^2].
\end{equation}

The intensity measure of the original process of power values $\Pi$ induced by $\phi$ is obtained by replacing $t$ with $1/t'$ in (\ref{intensity1}),  giving
\OneOrTwoColumnDisplay{
\begin{align}\label{intensity3}
\bar{\Measure}(\tup) :=
\bar{\Measure}([t',\infty))=\E[\sum_{P_i\in \Pi}\Ind(P_i \ge t')]
 {=\E[\sum_{V_i\in N}\Ind(V_i \le 1/t')]} =\E [\phi (  h^{-1} ( \fade /t'))] .
\end{align}
}{
\begin{align}\label{intensity3}
\nonumber
\bar{\Measure}([t',\infty))&=\E[\sum_{P_i\in \Pi}\Ind(P_i \ge t')]\\ 
&\, {=\E[\sum_{V_i\in N}\Ind(V_i \le 1/t')]} =\E [\phi (  h^{-1} ( \fade /t'))] .
\end{align}
}
Similarly, for random $\Phi$  {of transmitters} with intensity measure $\Lambda$, expression (\ref{intensity2}) gives
\begin{equation}\label{intensity4}
 {\bar{\Measure}([t',\infty))=}\E[\Lambda(  h^{-1} (\fade/t'))].
\end{equation}  

If the transmitters form a Poisson point process $\Phi$ with intensity measure $\Lambda$, then the Poisson mapping theorem says that the process $N$ is a Poisson point process on the positive real line with intensity measure determined by~(\ref{intensity2}), while the process of power values $\Pi$ is also a Poisson process with intensity measure satisfying~(\ref{intensity4}). If $N$ is not induced by an underlying Poisson process of transmitters, but $N$ is still stochastically close to a Poisson point process with intensity measure $\Measure$, then the process $N$ is close to the process that is induced by transmitters placed according to a Poisson process, so one can assume transmitter locations form a Poisson process.  The error made in this substitution can be quantified~\cite[Theorem~2.7]{keeler2016wireless}.

\subsection{Examples of $\Measure(t)$}
The standard path-loss model is $\ell(x)=|x|^{-\beta}$, where $\beta>2$, hence $h^{-1}(y)=y^{1/\beta}$, $y>0$.  If the transmitters form a stationary point process $\Phi$ on $\R^2$ with density $\lambda$, the resulting intensity measure satisfies $\Measure(t)=\lambda  \pi t^{2/\beta} \E(\fade^{2/\beta})$, which depends on $S$ through only one moment $\E(\fade^{2/\beta})$. For Poisson $\Phi$, one can assume $S$ is, for example, exponential, perform calculations, and then remove the exponential assumption and change to another model of $S$ by rescaling $M$. Writing $r=\vert x \vert$, the multi-slope model is
\OneOrTwoColumnDisplay{
$
h(r)
=\sum_{i=1}^{{k+1}}  b_i^{{-1}} r^{\beta_i}   \ind({r_{i-1}\leq r <r_i}) ,
$
}{
$$
h(r)
=\sum_{i=1}^{{k+1}}  b_i^{{-1}} r^{\beta_i}   \ind({r_{i-1}\leq r <r_i}) ,
$$
}
where $0={r_0<}r_1<\cdots< r_k<r_{k+1}=\infty$, $\beta_i>0$, and $b_i>0$ are set appropriately so  $h$ is continuous. Each interval $[r_{i-1},r_i)$ is disjoint, so the inverse of $h(r)$ is
$
h^{-1}(s)=\sum_{i=1}^{{k+1}}  c_i s^{1/\beta_i}  \ind({s_{i-1}\leq s   <s_i }) ,
$
where $s_i=b_i^{{-1}} r_i^{\beta_i}$ {and $c_i=b_i^{1/\beta_i}$.}  For stationary $\Phi$,
\OneOrTwoColumnDisplay{
$
\Measure(t)=
2\pi\lambda \sum_{i=1}^k t^{2/\beta_i} c_i \E\left[ {  {\fade}^{2/\beta_i} \ind({s_{i-1}\leq t {\fade} <s_i}})  \right]. 
$
}{
$$
\Measure(t)=
2\pi\lambda \sum_{i=1}^k t^{2/\beta_i} c_i \E\left[ {  {\fade}^{2/\beta_i} \ind({s_{i-1}\leq t {\fade} <s_i}})  \right]. 
$$}
Care must be taken when determining the generalized inverse $h^{-1} $. For example, the function $\ell(x)=e^{-\beta |x|}$ gives $h^{-1}(y)=(1/\beta)\ln^+(y):=(1/\beta)\max[0, \ln (y)]$ for $y \geq 0$, and not $(1/\beta)\ln(y)$. For stationary $\Phi$, 
$
\Measure(t)= \frac{\lambda\pi}{\beta^2} \mean([ \ln^+(t\fade) ]^2).
$
If $\fade$ is continuous  on $[0, \infty)$ with probability density $f_{\fade}$, then
$
\Measure(t)= \frac {\lambda\pi} { t \beta^2 } \int _1^{\infty} (\ln x)^2 f_{\fade}(x/t) dx.
$

\subsection{Approximating signals with a Poisson process}
For each transmitter $x_i\in\phi$ (deterministic with arbitrary indexing), let  $p_{x_i}(t):= \Prob(0< 1/(\ell(x_i)\fade_i)\leq t)=\Prob(0< V_i \leq t)$, 
we want to approximate the point process $N=\{V_i\}_{i\geq1}$ with a Poisson point process $Z=\{Y_i\}_{i\geq1}$ with intensity measure $\Measure$, given by (\ref{intensity1}) or, equivalently, $M(t)=\sum_{x_i \in\phi} p_{x_i}(t)$, so we need to introduce a probability metric. For two probability measures $\mu$ and $\nu$ defined on the same probability space with $\sigma$-algebra $\mathcal{F}$ of events, the total variation distance is $\dtv(\mu,\nu): = \sup_{A\in \mathcal{F}}|\mu(A)  -\nu(A)|$, which is a strong metric that bounds the largest difference in probabilities between two distributions. We write $\law(U)$ to denote the distribution or law of a point process $U$ (or other random objects).  The two point processes $Z$ and $N$ are stochastically close if their laws $\law(Z)$ and $\law(N)$ have a small total variation $\dtv(\law(Z), \law(N))$. 
To obtain meaningful values for the total variation, one must compare the point processes on a finite interval (it's of no practical use to  compare infinite configurations of points because the difference of undetectable weak signals will dictate the total variation distance), which is  a slight but necessary restriction.  {We present an approximation theorem~\cite[Theorem~2.2]{keeler2016wireless} for the restricted parts of $Z$ and $N$.}
\begin{Theorem}\label{thmdtv}  {For $\tt>0$, let $N\res\tt$ and $Z\res\tt$ be the points of the point processes $N$ and $Z$ restricted to the interval $(0,\tt]$.  Then}
\OneOrTwoColumnDisplay{
\begin{align}
\nonumber \frac{1}{32}\min[1,& 1/\Measure(\tt)] \sum_{x_i\in \phi} p_{x_i}(\tt)^2  \leq
\dtv(\law(Z\res \tt), \law(N\res \tt))   \leq \sum_{x_i\in \phi} p_{x_i}(\tt)^2  \leq \Measure(\tt) \max_{x_i\in \phi} p_{x_i}(\tt) .
\end{align}
}{
\begin{align}
\nonumber \frac{1}{32}\min[1,& 1/\Measure(\tt)] \sum_{x_i\in \phi} p_{x_i}(\tt)^2  \leq
\dtv(\law(Z\res \tt), \law(N\res \tt))  \\ \leq &\sum_{x_i\in \phi} p_{x_i}(\tt)^2  \leq \Measure(\tt) \max_{x_i\in \phi} p_{x_i}(\tt) .
\end{align}
}

\end{Theorem}
 {Note that} since the theorem is about the inverse signal strengths,  {$1/\tt$ can be interpreted as the smallest possible power value of interest for an observer in the network and $\Measure(\tt)$ as the expected number of signals with power value}  {greater than or equal to}  {$1/\tt$.} Essentially the theorem says that if the network has many transmitters with independent signal strengths, and the chance
is small that any particular transmitter has  {signal strength more powerful than~$1/\tt$}, then the point process of signal strengths  should be close to Poisson with the same intensity measure.  {If our interest is on the strongest $\Measure(\tt)=3$, say, signals, then, as we increase the number of transmitters, $\tt$ decreases, hence $p_{x_i}(\tt)$ and the bounds decrease, meaning that these signal strengths are closer to a Poisson process.}
If  $x_*$ is the transmitter closest to the origin (with fading $S_*$), then $\max_{x_i\in \phi} p_{x_i}(\tt) =\Prob(0< h(x_*)/\fade_*  \leq \tt)$, which is not as tight as the $\sum_{x_i\in \phi} p_{x_i}(\tt)^2$ term but is often easier to calculate; we examine the theorem further in Section~\ref{s.theoremanalysis}. 

\subsection{Poisson convergence}
Theorem~\ref{thmdtv} is an approximation result, but it has been used to show Poisson convergence of the process $N$ on the whole positive line~\cite[Theorem 1.1]{keeler2016wireless}. 
Let $\convp$ denote convergence in probability and $L(t)$ be a non-decreasing
function on $\R_0^+$, which induces an intensity measure of a Poisson process.  

\begin{Theorem}\label{thmshadow}
Assume 
\begin{equation}\label{lambda}
\lim\limits_{r\rightarrow\infty}\frac{\phi(r)}{\pi r^2} = \lambda,
\end{equation}
 so that the transmitters have a nearly constant density. Let $(\fade(v))_{v\geq0}$ be a family of positive random variables indexed by some non-negative parameter~$v$, $N\t v$ be the point
process generated by $\fade(v)$, $g$ and $\phi$. If as $v\to\infty$, 
(i) $\fade(v)\convp0$ and (ii) $ \Measure^{(v)}(t):= \mean [\phi(h^{-1}(\fade(v) t)] \to L(t)$, where $t>0$, then $N\t v$  converges weakly to a Poisson process {on $\R_+^0$} with intensity measure $L$.
\end{Theorem}

The parameter $v$ can be any parameter of the distributions $\fade(v)$.  
To give an example, we assume $\ell= \vert x\vert^{-\beta}$, a transmitter configuration $\phi$ that meets condition~(\ref{lambda}), and  $\{ \shadow^{(v)}_i\}_{i\geq 1}$ are iid log-normal variables, such that $\shadow^{(v)}_i=\exp[v B_i-v^2/\beta]$, where each $B_i$ is a standard normal variable, so $E[(\shadow^{(v)}_i)^{2/\beta}]=1$. Then  {for any $s>0$,} $\Prob(\, B_i> \log(s)/v +v/\beta\, )\to 0$ or $\shadow^{(v)}\convp0 $  as $v\rightarrow \infty$, which is condition (i) in Theorem~\ref{thmshadow}, and $\Measure^{(v)}(t)=\lambda\pi t^{2/\beta}$, so we recover~\cite[Theorem~7]{blaszczyszyn2014wireless}.

We consider a Rayleigh model, where $\{ \rayleigh^{(v)}_i\}_{i\geq 1}$ are iid exponential variables with mean $1/v$.   {For any $s>0$,} $\Prob(\, \rayleigh^{(v)}>s\,)=e^{-v s}$, then $ \rayleigh^{(v)}\convp0 $  as $v\rightarrow \infty$, so condition (i) is met again. But  $\Measure^{(v)}(t) \rightarrow 0 $ as $v\rightarrow\infty$, due to $\E([\rayleigh^{(v)}]^{2/\beta})= \Gamma(2/\beta+1)/v^{2/\beta} $. In the limit as $v\rightarrow\infty$, the Rayleigh model does not give a meaningful $L$. But under this and other models, the process $N^{(v)}$ can still be approximated with a Poisson process with intensity measure $\Measure^{(v)}$ for sufficiently large $v$. 


If we replace the non-random $\phi$ with a random point process $\Phi$, all the above results hold with suitable modifications. But in the limit as $\fade(v)\convp0$, it is possible to see a Cox process (a Poisson process with a random intensity measure), due to the extra randomness from $\Phi$.  To obtain a Poisson process in this limit, $\Phi$ must meet certain conditions, which are satisfied by, for example, the Ginibre process~\cite[Section~2.1]{keeler2016wireless}.  These results further demonstrate that in the presence of strong random propagation effects the  {signal strengths can appear as a Poisson or Cox process, whereas the transmitters form some other point process. }

\section{New results}
\subsection{Analysis of Theorem~\ref{thmdtv}}\label{s.theoremanalysis}
The upper bound in Theorem~\ref{thmdtv} is an increasing function in $\tau$. The smaller we make our interval $(0,\tt]$, where we only consider signals of power values greater than $1/\tt$, then the smaller the bounds and so the more Poisson the signals in $(0,\tt]$ behave. In other words, the stronger signals are stochastically the more Poisson ones. We believe that this is the first appearance of this observation, by purely probabilistic arguments, but it is supported by recent work on fitting a Poisson model to a real cellular phone network~\cite[Figure 6]{blaszczyszyn2015spatial}. A functional that is dependent on $k$ strongest signals can be well-approximated with a functional of a Poisson process with intensity $M(t)$. Arguably, the first few strongest signals  matter the most in the calculations of the signal-to-interference ratio, further explaining why the Poisson process has been a good model in practice
To calculate this ratio, Haenggi and Ganti~\cite{gantisir} have shown for a power-law path-loss model that Poisson network models can be used to approximate network models based on other stationary point processes, which may be connected to our Poisson approximation results. 

If $\tt$ is made too small, then the bounds lose meaning as there will be no signals in  $(0,\tt]$. Conversely, for large $\tt$ the upper bound will also lack meaning as it will be greater than one. What remains to be explored is for which values of $\tt$, under suitable propagation models, give meaningful values for the bounds in Theorem~\ref{thmdtv}. The theorem can be adapted easily for the original process of power values $\Pi$~\cite[Remark 2.3]{keeler2016wireless}.

\subsection{Statistics of the strongest signals}
An important feature of the total variation distance is that the bounds in Theorem~\ref{thmdtv}  will hold under simple functions of the two truncated  point processes $N\res\tt$ and $Z\res\tt$. We leverage this fact and the coupling interpretation of total variation distance to derive new bounds for the order statistics of the whole processes $N$ and $Z$. Let $V_{(1)}\leq  V_{(2)}  \leq \dots$ denote the increasing order statistics of the  process $N=\{V_i\}_{i\geq1}$. Similarly, let $Y_{(1)}\leq  Y_{(2)}  \leq \dots$ be the order statistics of the Poisson process $Z=\{Y_i\}_{i\geq1}$. 
\begin{Theorem}\label{thmorder}
For any $\tau>0$,
\OneOrTwoColumnDisplay{
\begin{align*}
\dtv(\law(V_{(1)},\ldots, V_{(k)}),  \law(Y_{(1)},\ldots,Y_{(k)}) )  \leq   \sum_{x_i\in \phi} p_{x_i}(\tt)^2+ \sum_{j=0}^{k-1} \frac{\Measure(\tau)^j e^{-\Measure(\tau)}}{j!}.
\end{align*}
}{
\begin{align*}
\dtv&(\law(V_{(1)},\ldots, V_{(k)}),  \law(Y_{(1)},\ldots,Y_{(k)}) ) \\ & \leq   \sum_{x_i\in \phi} p_{x_i}(\tt)^2+ \sum_{j=0}^{k-1} \frac{\Measure(\tau)^j e^{-\Measure(\tau)}}{j!}.
\end{align*}
}
\end{Theorem}
\begin{proof}
An equivalent definition of total variation distance between probability measures $\mu, \nu$ is  {\cite[p.~254]{BHJ92}}
\OneOrTwoColumnDisplay{
$
\dtv(\mu, \nu)=\min_{U\sim \mu, W\sim \nu} \Prob(U \not= W),
$
}{
$$
\dtv(\mu, \nu)=\min_{U\sim \mu, W\sim \nu} \Prob(U \not= W),
$$
}
where the minimum is taken over all couplings of $\mu$ and $\nu$. Thus Theorem~\ref{thmdtv} implies that there is a coupling 
$(\tilde Z \res \tau,\tilde N \res \tau)$
of $Z\res \tau$ and $N\res \tau$ such that $\Prob(\tilde Z\res \tau\not=\tilde N\res \tau)\leq \sum_{x_i\in \phi} p_{x_i}(\tt)^2$. We define a 
coupling of $(V_{(1)},\ldots, V_{(k)})$ and $(Y_{(1)},\ldots, Y_{(k)})$, by taking the first $k$ order statistics of $\tilde Z \res \tau$
and $\tilde N \res \tau$ if there are at least $k$ points in $\tilde Z \res \tau$ and  {$\tilde N \res \tau$,} otherwise we set the two vectors in the coupling to be independent.
For this coupling, the probability that the two vectors are not equal is upper bounded by
\OneOrTwoColumnDisplay{
$
\Prob(\tilde Z\res \tau\not=  \tilde N\res \tau)+\Prob(\#(\tilde Z\res \tau)\leq k-1) \leq   \sum_{x_i\in \phi} p_{x_i}(\tt)^2+ \sum_{j=0}^{k-1} \frac{\Measure(\tau)^j e^{-\Measure(\tau)}}{j!}. \qedhere
$
}{
\begin{align*}
\Prob(\tilde Z\res \tau\not= & \tilde N\res \tau)+\Prob(\#(\tilde Z\res \tau)\leq k-1) \\&\leq   \sum_{x_i\in \phi} p_{x_i}(\tt)^2+ \sum_{j=0}^{k-1} \frac{\Measure(\tau)^j e^{-\Measure(\tau)}}{j!}. \qedhere
\end{align*}
}
\end{proof}

To interpret the theorem, note that the second term in the bound can be made small by choosing $\tau$ large and then for fixed $\tau$ the first term is small if each fading variable $\fade_i$ is small with good probability. As $\tau$ grows, the second part of the bound improves while the first part worsens. This reflects the trade-off between the strongest signals being nearly a Poisson process and the bounds needing to apply to the order statistics of the whole process. 
For a Poisson $N$, the distribution of $V_{(1)}$ is 
$
\Prob( V_{(1)} \leq t) = 1-e^{-\Measure(t)},
$
which can approximate the  distribution  of $Y_{(1)}$. 
If  $V_{(1)}$  and $Y_{(1)}$ are continuous on $[0, \infty)$ with probability densities $f_V$ and $f_Y$,  then
$
 \dtv(\law(V_{(1)}), \law(Y_{(1)}) )  = \frac{1}{2}\int_0^{\infty}|f_V(x)-f_Y(x)|dx .
$

\subsection{Poisson convergence of strongest signals}
Theorem~\ref{thmorder} leads to a Poisson convergence result. 
\begin{Theorem}\label{thmshadoworder}
Under the notation and assumptions of Theorem~\ref{thmshadow},
for $i\geq 1$ {, let $V\t v_{(i)}$ and $Y_{(i)}$ be the $i$th smallest order statistic of the process $N\t v$ 
and a Poisson process on $\R_+^0$
with intensity measure~$L$ respectively}. Then for fixed $k\geq1$ and as $v\to\infty$,
\OneOrTwoColumnDisplay{
$
 {\law(V\t v_{(1)},\ldots, V\t v_{(k)})}\to \law(Y_{(1)},\ldots,Y_{(k)}).
$
}{
\[
 {\law(V\t v_{(1)},\ldots, V\t v_{(k)})}\to \law(Y_{(1)},\ldots,Y_{(k)}).
\]
}
\end{Theorem}
\begin{IEEEproof}
Let  {$Y\t v_{(i)}$} be the  {$i$th smallest} order statistic of a 
Poisson process with intensity measure~$M\t v(t)$.
Then due to the convergence of mean measures, 
 {$\law(Y\t v_{(1)},\ldots,Y\t v_{(k)})\to\law(Y_{(1)},\ldots,Y_{(k)})$}. We now use Theorem~\ref{thmorder}
to show the total variation distance between the $V\t v_{(i)}$'s and $Y\t v_{(i)}$'s tends 
to zero. For fixed $\tau>0$, it was shown 
in \cite{keeler2016wireless} (see Theorems~1.1 and Corollary~2.5) that 
under the conditions of the theorem the first term from the bound of Theorem~\ref{thmorder} tends 
to zero as $v\to\infty$. Thus for all  {$\tau> 0$, with $L$ in Theorem~\ref{thmshadow}},
\OneOrTwoColumnDisplay{
$
\limsup_{v\to\infty}\dtv( {\law(V\t v_{(1)},\ldots, V\t v_{(k)}),  \law(Y\t v_{(1)},\ldots,Y\t v_{(k)}) )}
\leq \sum_{j=0}^{k-1} \frac{L(\tau)^j e^{-L(\tau)}}{j!},
$}{
\begin{align*}
&\limsup_{v\to\infty}\dtv( {\law(V\t v_{(1)},\ldots, V\t v_{(k)}),  \law(Y\t v_{(1)},\ldots,Y\t v_{(k)}) )}\\
&\qquad\leq \sum_{j=0}^{k-1} \frac{L(\tau)^j e^{-L(\tau)}}{j!},
\end{align*}
}
and the proof is completed by sending $\tau\to\infty$ (note that  {$\lim_{\tau\to\infty}L(\tau)=\infty$} since $\phi$ is infinite).
\end{IEEEproof}


\subsection{Estimating $\Measure(t)$}\label{s.estimateM}
The empirical distribution of $ Y_{(1)}$, denoted by $\hat{E}(t)$, gives a way to statistically estimate or fit $\Measure(t)$ by first assuming that the transmitters are positioned according to a Poisson model, even if the transmitters don't appear Poisson, and then approximating with $\Prob( V_{(1)} \leq t)$.   One measures the largest signal in different locations and  fits the empirical distribution of $ Y_{(1)}$  to the equation: 
$
-\log[1-\hat{E}(t)] = \hat{\Measure(t)} ,\ t>0 ,
$
where $\hat{\Measure(t)}$ is the estimate of $M(t)$. 

For a large cellular phone network, the intensity measure of $N$ has been fitted to experimental signal data under a Poisson network model with $\ell(x)=|x|^{-\beta}$~\cite{blaszczyszyn2014wireless}; also see~\cite{blaszczyszyn2015spatial} for models with antenna patterns. Recent work~\cite{lu2015stochastic} studies the intensity measure of the path-loss process (so the process $N$ with all $\{S_i\}_{i\geq 1}$ set to some constant) by using an advanced path-loss model, which can be used in our setting,  and geographic data from cellular networks in two cities. It is not remarked that this first-moment approach would hold for any stationary point process with density $\lambda>0$, where the intensity measure of the path-loss process is given by (\ref{intensity2a}), but if  random  $\fade$ was incorporated into the model, our results suggest that the Poisson model would be the most appropriate.

\section{Conclusion}
We justified the Poisson approximation  process in a general setting and derived results on the order statistics of the signal strengths. Similar results hold for other functions of  {signal} strengths if the function is relatively well-behaved around the origin (or at infinity for the inverse signal strengths). 
An interesting observation is that the stronger signals behave more Poisson, which is convenient for statistics that only depend on the strongest signals. For a single-observer perspective of a wireless network, our results suggest that the focus should not be on the locations of the transmitters, but rather studying the process of (inverse) power values on the real line. We described a simple procedure for statistically estimating the intensity measure of this process. 
We encourage future work to explore the practical validity of these results.  For example, independence may not be a justifiable assumption for (large-scale) shadowing. But for localized shadowing dependence, we expect similar results  {presented} here, since the kind of Poisson approximation our results rely on typically allows for such  dependence. In summary, if the network has many transmitters, the chance of any transmitter having a strong signal is small, and the dependence between signals of transmitters is localized, then the point process of signal strengths is close to a Poisson point process with the same intensity measure. If these assumptions are justifiable, then practitioners can estimate the intensity measure from empirical data and draw inferences under Poisson probabilities.

{\small
\pdfbookmark[0]{References}{References} 
\bibliographystyle{IEEEtran}
\bibliography{Convergence}
}

\end{document}